\documentclass[11pt]{article}

\usepackage[top=2cm, bottom=2cm, left=2cm, right=2cm, includefoot]{geometry}
\usepackage{authblk}
\usepackage{amsthm}
\usepackage{xcolor}
\usepackage{graphicx}
\usepackage{amsmath}
\usepackage[ruled,linesnumbered,algo2e,vlined]{algorithm2e}
\usepackage{amssymb}
\usepackage{color}
\usepackage{dcolumn}
\usepackage{tabularx}
\usepackage{url}
\usepackage[numbers,sort]{natbib}

\newtheorem{theorem}{Theorem}
\newtheorem{fact}{Fact}
\theoremstyle{definition}
\newtheorem{example}{Example}

\newcommand{\distq}{DIST$q$}
\newcommand{\ldistq}{LDIST$q$}

\newcommand{\kmpshift}{\mathit{KMP\_Shift}}

\newcommand{\border}{\mathit{Bord}}
\newcommand{\sborder}{\mathit{Strong\_Bord}}
\newcommand{\PreKMP}{\mathtt{PreKMPShift}}

\newcommand{\PreShift}{\mathtt{PreHqShift}}
\newcommand{\PreDarray}{\mathtt{PreDistArray}}

\newcommand{\darray}{\mathit{dist}}
\newcommand{\shift}{\mathit{HQ\_Shift}}

\newcommand{\moji}[1]{\texttt{#1}}
\newcommand{\batu}{\textcolor{red}{\times}}
\newcommand{\maru}{\textcolor{blue}{\circ}}
\newcommand{\pmaru}{\textcolor{blue}{\bullet}}

\newcommand{\Fib}{\mathit{Fib}}

\makeatletter
\newcolumntype{B}{>{\boldmath\DC@{.}{.}{-1}}c<{\DC@end}}
\makeatother

\def\HiLi{\leavevmode\rlap{\hbox to \hsize{\color{yellow!50}\leaders\hrule height .8\baselineskip depth .5ex\hfill}}}

\begin{document}

\title{Fast and linear-time string matching algorithms based on the distances of $q$-gram occurrences}
\author[]{Satoshi Kobayashi}
\author[]{Diptarama Hendrian}
\author[]{Ryo Yoshinaka}
\author[]{Ayumi Shinohara}

\affil{Graduate School of Information Sciences, Tohoku University, Japan}

\date{}

\maketitle            

\begin{abstract}
Given a text $T$ of length $n$ and a pattern $P$ of length $m$, the string matching problem is a task to find all occurrences of $P$ in $T$. In this study, we propose an algorithm that solves this problem in $O((n + m)q)$ time considering the distance between two adjacent occurrences of the same $q$-gram contained in $P$. We also propose a theoretical improvement of it which runs in $O(n + m)$ time, though it is not necessarily faster in practice. We compare the execution times of our and existing algorithms on various kinds of real and artificial datasets such as an English text, a genome sequence and a Fibonacci string. The experimental results show that our algorithm is as fast as the state-of-the-art algorithms in many cases, particularly when a pattern frequently appears in a text.
\end{abstract}

\section{Introduction}\label{sec:introduction}
The exact string matching problem is a task to find all occurrences of $P$ in $T$ when given a text $T$ of length $n$ and a pattern $P$ of length $m$.
A brute-force solution of this problem is to compare $P$ with all the substrings of $T$ of length $m$.
It takes $O(nm)$ time.
The Knuth-Morris-Pratt (KMP) algorithm~\cite{Knuth1977kmp} is well known as an algorithm that can solve the problem in $O(n+m)$ time.
However, it is not efficient in practice, because it scans every position of the text at least once. 
The Boyer-Moore algorithm~\cite{Boyer1977bm} is famous as an algorithm that can perform string matching fast in practice by skipping many positions of the text, though it has $O(nm)$ worst-case time complexity.
Like this, many efficient algorithms whose worst-case time complexity is the same or even worse than the naive method have been proposed so far~\cite{Horspool1980hor, Sunday1990qs, Lecroq2007hashq}.
For example, the HASH$q$ algorithm~\cite{Lecroq2007hashq} focuses on the substrings of length $q$ in a pattern and obtains a larger shift amount.
However, considering that such an algorithm is embedded in software and actually used, if the worst-case input strings are given, the operation of the software may be slowed down.
Therefore, an algorithm that operates theoretically and practically fast is important.
Franek et al.~\cite{Franek2007fjs} proposed the Franek-Jennings-Smyth (FJS) algorithm, which is a hybrid of the KMP algorithm and the Sunday algorithm~\cite{Sunday1990qs}.
The worst-case time complexity of the FJS algorithm is $O(n+m+\sigma)$ and it works fast in practice, where $\sigma$ is the alphabet size.
Kobayashi et al.~\cite{Kobayashi2019improvedfjs} proposed an algorithm that improves the speed of the FJS algorithm by combining a method that extends the idea of the Quite-Naive algorithm~\cite{Cantone2004qn}.
This algorithm has the same worst-case time complexity as the FJS algorithm, and it runs faster than the FJS algorithm in many cases.
The LWFR$q$ algorithm~\cite{Cantone2019wfr} is a practically fast algorithm that works in linear time.
This algorithm uses a method of quickly recognizing substrings of a pattern using a hash function.
See~\cite{FaroLecroqSurvey2013,Hakak2019survey} for recent surveys on exact string matching algorithms.

This paper proposes two new exact string matching algorithms based on the HASH$q$ and KMP algorithms incorporating a new idea based on the distances of occurrences of the same $q$-grams.
The time complexity of the preprocessing phase of the first algorithm is $O(mq)$ and the search phase runs in $O(nq)$ time.
The second algorithm improves the theoretical complexity of the first algorithm, where the preprocessing and searching times are $O(m)$ and $O(n)$, respectively.
Our algorithms are as fast as the state-of-the-art algorithms in many cases.
Particularly, our algorithms work faster when a pattern frequently appears in a text.

This paper is organized as follows. 
Section~\ref{sec:preliminaries} briefly reviews the KMP and HASH$q$ algorithms, which are the basis of the proposed algorithms.
Section~\ref{sec:proposed} proposes our algorithms.
Section~\ref{sec:experiments} shows experimental results comparing the proposed algorithms with several other algorithms using artificial and practical data.
Section~\ref{sec:conclusion} draws our conclusions.

\section{Preliminaries}\label{sec:preliminaries}

\subsection{Notation}
Let $\Sigma$ be a set of characters called an \emph{alphabet} and $\sigma = |\Sigma|$ be its size.
$\Sigma^*$ denotes the set of all strings over $\Sigma$.
The length of a string $w \in \Sigma^*$ is denoted by $|w|$.
The \emph{empty string}, denoted by $\varepsilon$, is the string of length zero.
The $i$-th character of $w$ is denoted by $w[i]$ for each $1 \le i \le |w|$.
The substring of $w$ starting at $i$ and ending at $j$ is denoted by $w[i:j]$ for $1 \le i \leq j \le |w|$.
For convenience, let $w[i:j] = \varepsilon$ if $i > j$.
A string $w[1:i]$ is called a \emph{prefix} of $w$
and a string $w[i:|w|]$ is called a \emph{suffix} of $w$.
A string $v$ is a \emph{border} of $w$ if $v$ is both a prefix and a suffix of $w$.
Note that the empty string is a border of any string.
Moreover, a prefix, a suffix or a border $v$ of $w$ is called \emph{proper} when $v \neq w$.
The length of the longest proper border of $w[1:i]$ for $1 \le i \le |w|$ is given by
\begin{align*}
	\border_w[i] = 
\max\{\,j \mid w[1:j] = w[i-j+1:i] \text{ and } 0 \le j < i\,\}\,.
\end{align*}
Throughout this paper, we assume $\Sigma$ is an integer alphabet.

\subsection{The exact string matching problem}
The exact string matching problem is defined as follows:
\begin{align*}
	\textbf{Input: }& \text{A text $T \in \Sigma^*$ of length $n$ and a pattern $P \in \Sigma^*$ of length $m$,} \\
	\textbf{Output: }& \text{All positions $i$ such that $T[i:i+m-1]=P$ for $1 \leq i \leq n-m+1$.}
\end{align*}
We will use a text $T \in \Sigma^*$ of length $n$ and a pattern $P \in \Sigma^*$ of length $m$ throughout the paper.

Let us consider comparing $T[i:i+m-1]$ and $P[1:m]$.
The naive method compares characters of the two strings from left to right.
When a character mismatch occurs, the pattern is shifted to the right by one character.
That is, we compare $T[i+1:i+m]$ and $P[1:m]$.
This naive method takes $O(nm)$ time for matching. 
There are a number of ideas to shift the pattern more so that searching $T$ for $P$ can be performed more quickly,
using shifting functions obtained by preprocessing the pattern.

\subsection{Knuth-Morris-Pratt algorithm}
The Knuth-Morris-Pratt (KMP) algorithm~\cite{Knuth1977kmp} is well known as a string matching algorithm that has linear worst-case time complexity.
When the KMP algorithm has confirmed that $T[i:i+j-2] = P[1:j-1]$ and $T[i+j-1] \neq P[j]$ for some $j \leq m$, it shifts the pattern so that a suffix of $T[i:i+j-2]$ matches a prefix of $P$ and we do not have to re-scan any part of $T[i:i+j-2]$ again.
That is, the pattern can be shifted by $j-k-1$ for $k=\border_P[j-1]$.
In addition, if $P[k+1] = P[j]$, the same mismatch will occur again after the shift.
In order to avoid this kind of mismatch, we use $\sborder[1:m+1]$ given by
\begin{equation*}
	\sborder_P(j) = \left\{
		\begin{aligned}
			& \border_P(m) && \hspace{-3mm}\text{if $j = m + 1$,} \\
			& \max\!\big(\{\,k \mid P[1:k] = P[j-k:j-1],\; P[k+1] \neq P[j],\; \\ 
			& \text{\phantom{$\max\big(\{\,k \mid {}$}} 0 \le k < j\,\} \cup \{-1\}\big) && \hspace{-3mm}\text{otherwise.} \\
		\end{aligned}
	\right.
\end{equation*}
The amount $\kmpshift[j]$ of the shift is given by
\begin{align*}
	\kmpshift[j] = j - \sborder_P(j) - 1.
\end{align*}
This function has a domain of $\{1,\dots,m+1\}$ and is implemented as an array in the algorithm.
Hereafter, we identify some functions and the arrays that implement them.
\begin{fact}\label{fact:kmp}
If $P[1:j-1]=T[i:i+j-2]$ and $P[j] \neq T[i+j-1]$, then $P[1:j-k_j-1] = T[i+k_j:i+j-2]$ holds for $k_j = \kmpshift[j]$.
Moreover, there is no positive integer $k < \kmpshift[j]$ such that $P = T[i+k:i+k+m-1]$.
\end{fact}
Note that if the algorithm has confirmed $T[i:i+m-1] = P$, the shift is given by $\kmpshift[m+1]$ after reporting the occurrence of the pattern.
Algorithm~\ref{alg:kmp_kmpshift} shows a pseudocode to compute the array $\kmpshift$.
It runs in $O(m)$ time.
By using $\kmpshift$, the KMP algorithm finds all occurrences of $P$ in $T$ in $O(n)$ time.
\begin{algorithm2e}[btp]
\selectfont
\caption{Computing $\kmpshift$}
\label{alg:kmp_kmpshift}
\SetKwProg{Fun}{Function}{}{end}
\Fun{$\PreKMP(P)$}{
	$m \leftarrow |P|$;
	$i \leftarrow 1$;
	$j \leftarrow 0$\;
	$\sborder_P[1] \leftarrow -1$\;
	\While {$i \leq m$}{
		\lWhile{$j > 0$ \textup{and} $P[i] \neq P[j]$}{%
			$j \leftarrow \sborder_P[j]$%
		}
		$i \leftarrow i + 1$;
		$j \leftarrow j + 1$\;
		\lIf{$i \leq m$ \textup{and} $P[i] = P[j]$}{%
			$\sborder_P[i] \leftarrow \sborder_P[j]$%
		}\lElse{%
			$\sborder_P[i] \leftarrow j$%
		}
	}
	\lFor{$j \leftarrow 1$ \textup{\textbf{to}} $m$}{%
		$\kmpshift[j] \leftarrow j - \sborder_P[j] - 1$%
	}
	\Return $\kmpshift$
}
\end{algorithm2e}

\subsection{HASH$q$ algorithm}\label{sec:hashq}
The HASH$q$ algorithm~\cite{Lecroq2007hashq} is an adaptation of the Wu-Manber multiple string matching algorithm~\cite{Wu1994afast} to the single string matching problem.
Before comparing $P$ and $T[i:i+m-1]$, the HASH$q$ algorithm shifts the pattern so that the suffix $q$-gram $T[i+m-q:i+m-1]$ of the text substring shall match the rightmost occurrence of the same $q$-gram in the pattern.
For practical efficiency, we use a hash function, though it may result in aligning mismatching $q$-grams occasionally.
The shift amount is given by $\mathit{shift}(h(T[i+m-q:i+m-1]))$ where
\begin{align*}
\mathit{shift}(c) &= m - \max(\{j\;|\; h(P[j-q+1:j])=c ,\; q \leq j \leq m\} \cup \{q-1\}) \,,
\\
h(x) &= (2^{q-1} \cdot x[1] + 2^{q-2} \cdot x[2] + \cdots + 2 \cdot x[q-1] + x[q]) {\rm \;mod\; 2^8}.
\end{align*}
We repeatedly shift the pattern till the suffix $q$-grams of the pattern and the considered text substring have a matching hash value, in which case the shift amount will be 0.
We then compare the characters of the pattern and the text substring from left to right.
If a character mismatch occurs during the comparison, the pattern is shifted by
\begin{equation}\label{eq:hashq_dist}
 \min(\{\,k \mid h(P[m'-k:m-k])=h(P[m':m]) ,\; 1 \leq k \le m-q\,\} \cup \{m'\})
\end{equation}
where $m'=m-q+1$,
since the $q$-gram suffixes of the pattern and the text substring have the same hash values.
The time complexity of the preprocessing phase for computing the shift function is $O(mq)$.
The searching phase has $O(n(m+q))$ time complexity.
The worst-case time complexity is worse than that of the naive method, but it works fast in practice.
\begin{fact}\label{fact:hashq}
If $\mathit{shift}(h(T[i+m-q:i+m-1])) = j \neq m-q+1$, then $h(P[m-j-q+1:m-j])=h(T[i+m-q:i+m-1])$.
There is no positive integer $k < j$ such that $P = T[i+k:i+k+m-1]$.
\end{fact}

\section{Proposed algorithms}\label{sec:proposed}

\subsection{DIST$q$ algorithm}\label{sec:distq}
Our proposed algorithm uses three kinds of shifting functions.
The first one $\shift$ is essentially the same as $\mathit{shift}$, the one used in the HASH$q$ algorithm, except for the hashing function.
The second one $\darray$ is based on the distance of the closest occurrences of the $q$-grams of the same hash value in the pattern.
We involve $\kmpshift$ as the third one to guarantee the linear-time behavior.

Formally, the first shifting function is given as
\begin{align*}
\shift[c] &= m - \max(\{j \mid h(P[j-q+1:j])=c ,\; q \leq j \leq m\} \cup \{q-1\}),
\\ \text{where }\;& h(x) = (4^{q-1} \cdot x[1] + 4^{q-2} \cdot x[2] + \cdots + 4 \cdot x[q-1] + x[q]) {\rm \;mod\; 2^{16}}.
\end{align*}
Fact~\ref{fact:hashq} holds for $\shift$.

The second shifting function is defined for $j=q,\dots,m$ by
\begin{align*}
\darray[j] &= \min(\{\, k \mid h(P[j'-k:j-k]) = h(P[j':j]),\; 1 \le k \le j-q \,\} \cup \{j'\})
\end{align*}
where $j'=j-q+1$.
This function $\darray$ is a generalization of the shift (Eq. $\ref{eq:hashq_dist}$) used in the HASH$q$ algorithm.
We have $\darray[j] = k < j'$ if the $q$-gram ending at $j$ and the one ending at $j-k$ have the same hash value, while no $q$-grams occurring between those have the same value.
If no $q$-gram ending before $j$ has the same hash value, then $\darray[j]=j'$.
By using this, in the situation where $h(P[j-q+1:j]) = h(T[i+j-q:i+j-1])$, when a mismatch occurs anywhere between $T[i:i+m-1]$ and $P$, the pattern can be shifted by $\darray[j]$.
\begin{fact}\label{fact:dist}
Suppose that $h(P[j-q+1:j])=h(T[i+j-q:i+j-1])$.
Then $h(P[j-q+1-\darray[j]:j-\darray[j]])=h(T[i+j-q:i+j-1])$, unless $\darray[j] = j-q+1$.
Moreover, there is no positive integer $k < \darray[j]$ such that $P = T[i+k:i+k+m-1]$.
\end{fact}

Those functions $\shift$, $\darray$ and $\kmpshift$ are computed in the preprocessing phase.
Algorithms~\ref{alg:dist_compute_shift} and~\ref{alg:dist_compute_darray} compute the arrays $\shift$ and $\darray$, respectively.
\begin{algorithm2e}[t]
\selectfont
\caption{Computing $\shift$}
\label{alg:dist_compute_shift}
\SetKwProg{Fun}{Function}{}{end}
\Fun{$\PreShift(P, q)$}{
	$m \leftarrow |P|$\;
	\For {$i \leftarrow 0$ \textup{\textbf{to}} $2^{16}-1$} {
		$\shift[i] \leftarrow m - q + 1$\;
	}
	\For {$i \leftarrow q$ \textup{\textbf{to}} $m$} {
		$hash \leftarrow h(P[i-q+1:i])$\;
		$\shift[hash] \leftarrow m - i$\;
	}
	\Return $\shift$\;
}
\end{algorithm2e}

\begin{algorithm2e}[t]
\selectfont
\caption{Computing $\darray$}
\label{alg:dist_compute_darray}
\SetKwProg{Fun}{Function}{}{end}
\Fun{$\PreDarray(P, q)$}{
	\For{$j \leftarrow 1$ \textup{\textbf{to}} $q-1$}{
		$\darray[j] \leftarrow 1$\;
	}
	\For{$j \leftarrow 0$ \textup{\textbf{to}} $2^{16}-1$}{
		$prevpos[j] \leftarrow 0$\;
	}
	\For{$j \leftarrow q$ \textup{\textbf{to}} $|P|$}{
		$hash \leftarrow h(P[j-q+1:j])$\;
		\lIf{$prevpos[hash] = 0 $}{%
			$d \leftarrow j - q + 1$%
		}
		\lElse{%
			$d \leftarrow j - prevpos[hash]$%
		}
		$\darray[j] \leftarrow d$;
		$prevpos[hash] \leftarrow j$\;
	}
	\Return $\darray$\;
}
\end{algorithm2e}

Figure~\ref{fig:shift_d} shows examples of shifting the pattern using $\shift$ and $\darray$.
\begin{figure*}[t]
	\centering
	\includegraphics[width=13.5cm]{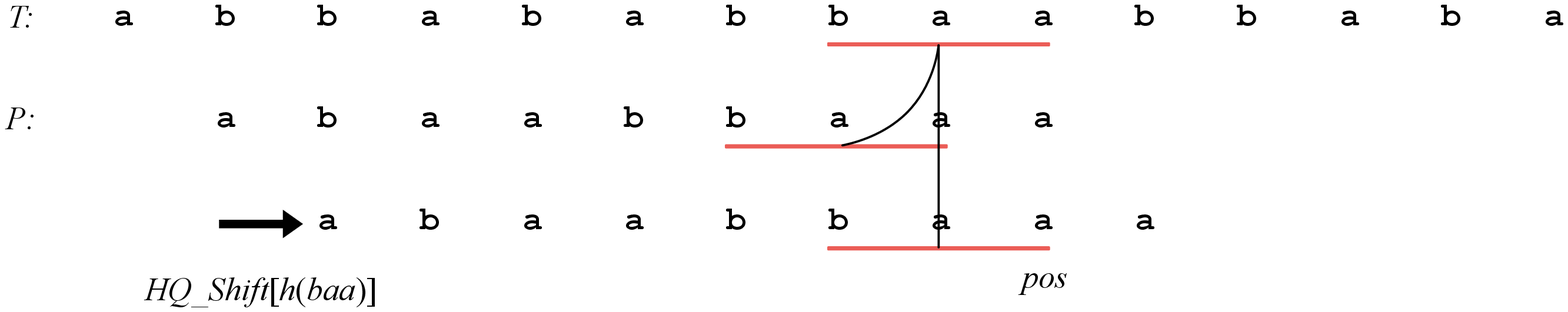}
	\includegraphics[width=13.5cm]{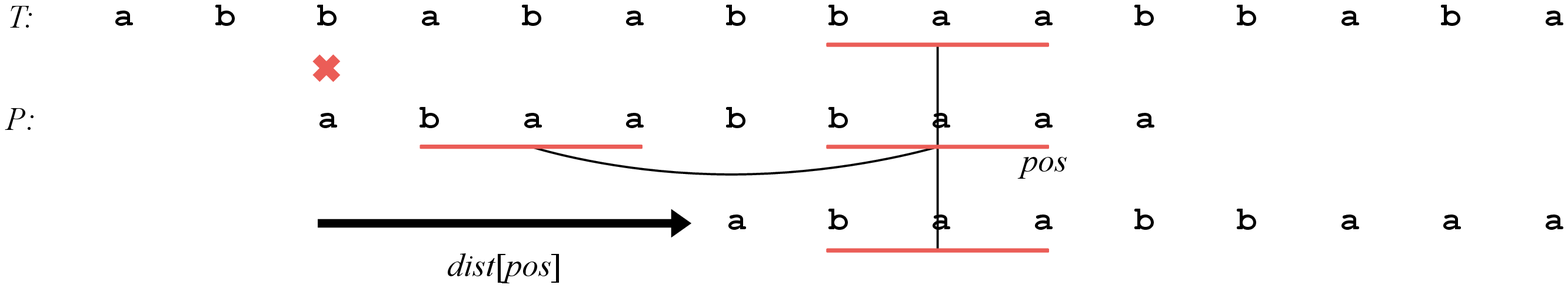}
	\caption{Shifting a pattern using $\shift$ and $\darray$}
	\label{fig:shift_d}
\end{figure*}
Both functions $\shift$ and $\darray$ shift the pattern using $q$-gram hash values based on Facts~\ref{fact:hashq} and~\ref{fact:dist}, respectively.
The latter can be used only when we know that the pattern and the text substring have aligned $q$-grams ending at $j$ with the same hash value
and it may shift the pattern at most $j-q+1$, while the former can be used anytime and the maximum possible shift is $m-q+1$.
The advantage of the function $\darray$ is in the computational cost.
If we know that the premise of Fact~\ref{fact:dist} is satisfied, we can immediately perform the shift based on $\darray$,
while computing $\shift(h(w))$ for the concerned $q$-gram $w$ in the text is not as cheap as $\darray[j]$.
Our algorithm exploits this advantage of the new shifting function $\darray$.

\begin{algorithm2e}[t!]
\selectfont
\caption{\distq{} algorihm}
\label{alg:dist_search_algorithm}
\SetKwProg{Fun}{Function}{}{end}
\Fun{$\mathtt{DISTq}(P,T, q)$}{
	$\kmpshift \leftarrow \PreKMP(P)$\;
	$\shift \leftarrow \PreShift(P, q)$\;
	$\darray \leftarrow \PreDarray(P, q)$\;
	$n \leftarrow |T|$;
	$m \leftarrow |P|$;
	$i \leftarrow 1$;
	$j \leftarrow 1$;
	$k \leftarrow m$\;
	\While{$k \leq n$}{
		\If{$j \leq 1$}{ 
			\While(\tcp*[f]{Alignment-phase}){{\rm True}}{
				$sh \leftarrow \shift[h(T[k-m+1:k])]$\label{ln:dist_hash};
				$k \leftarrow k + sh$\;
				\If{$sh \neq m - q + 1$}{
					$pos \leftarrow m - 1 - sh$\;
					\lIf{$P[1] = T[k-m+1]$}{%
						\textbf{break}%
					}
					$k \leftarrow k + \darray[pos]$\;
				}
				\lIf{$k > n$}{%
					\textbf{halt}%
				}
			}
			$j \leftarrow 2$;
			$i \leftarrow k-m+2$; \tcp*[f]{Comparison-phase}\\
			\While{$j \le m$ \textup{and} $P[j] = T[i]$}{
				$i \leftarrow i+1$;
				$j \leftarrow j+1$\;
			}
			\If{$j = m + 1$}{
				\textbf{output} $i-m$\;
			}
			
			\If{$\darray[pos] \geq j - 1 {\rm \;and\;} \darray[pos] \geq \kmpshift[j]$}{
				$j \leftarrow j - \darray[pos]$\;
			}
			\Else{
				$j \leftarrow j - \kmpshift[j]$\;
			}
		}\Else{
			\While(\tcp*[f]{KMP-phase}){$j \leq m$ \textup{and} $P[j] = T[i]$}{
				$i \leftarrow i+1$;
				$j \leftarrow j+1$\;
			}
			
			\If{$j = m + 1$}{
				\textbf{output} $i-m$\;
			}
			$j \leftarrow j - \kmpshift[j]$\label{ln:kmp2}\;
		}
		$k \leftarrow i+m-j$\;
	}
}
\end{algorithm2e}

Next, we explain our searching algorithm shown in Algorithm~\ref{alg:dist_search_algorithm}.
The searching phase is divided into three: \emph{Alignment-phase}, \emph{Comparison-phase}, and \emph{KMP-phase}.
The goal of the Alignment-phase is to shift the pattern as far as possible without comparing each single character of the pattern and the text.
The Alignment-phase ends when we align the pattern and a text substring that have (a) aligned $q$-grams of the same hash value and (b) the same first character.
Suppose $P$ and $T[k-m+1:k]$ are aligned at the beginning of the Alignment-phase.
If $s = \shift[h(T[k-q+1:k])] \le m-q$, by shifting the pattern by $s$, we find the aligned $q$-grams of the same hash value.
Namely, $h(P[m-q-s+1:m-s])=h(T[k-q+1:k])$.
Otherwise, we shift the pattern by $m-q+1$ repeatedly until we find aligned $q$-grams of the same hash value.
When finding a position $\mathit{pos}$ satisfying (a) by aligning $P$ and $T[k'-m+1:k']$ for some $k'$, i.e., $h(P[\mathit{pos}-q+1:\mathit{pos}])=h(T[k'-m+\mathit{pos}-q+1:k'-m+\mathit{pos}])$, we simply check the condition (b).
If $P[1]$ and the corresponding text character match, we move to the Comparison-phase.
Otherwise, we safely shift the pattern using $\darray[\mathit{pos}]$.
Note that although it is possible to use the function $\shift(h(T[k'-q+1:k']))$ rather than $\darray[\mathit{pos}]$, the computation would be more expensive.
Shifting the pattern by $\darray[\mathit{pos}]$, unless $\darray[\mathit{pos}] = \mathit{pos}-q+1$, still the pattern and the aligned text substring satisfy (a).
However, we do not repeat $\darray$-shift any more, since the smaller $\mathit{pos}$ becomes, the smaller the expected shift amount will be.
We simply restart the Alignment-phase.
Once the conditions (a) and (b) are satisfied, we move on to the Comparison-phase.

In the Comparison-phase, we check the characters from $P[2]$ to $P[m]$.
If a character mismatch occurs during the comparison, either of the shift by $\kmpshift$ or by $\darray$ is possible.
Therefore, we select the one where the resumption position of the character comparison goes further to the right after shifting the pattern.
If the resumption position of the comparison is the same, we select the one with the larger shift amount.
Recall that when the KMP algorithm finds that $P[1:j-1]=T[i:i+j-2]$ and $P[j] \neq T[i+j-1]$, 
it resumes comparison from checking the match between $T[i+j-1]$ and $P[j-\kmpshift[j]]$ if $\kmpshift[j] < j$, and $T[i+j]$ and $P[1]$ if $\kmpshift[j] = j$.
On the other hand, if we shift the pattern by $\darray[pos]$, we simply resume matching $T[i+\darray[pos]]$ and $P[1]$.
Therefore, we should use $\kmpshift[j]$ rather than $\darray[pos]$ when either $\kmpshift[j] < j$ and $\darray[pos] < j-1$ or $\kmpshift[j] = j > \darray[pos]$.
Summarizing the discussion, we shift the pattern by $\darray[pos]$ if $\darray[pos] \geq j-1$ and $\darray[pos] \geq \kmpshift[j]$ hold.
Otherwise, we shift the pattern by $\kmpshift[j]$.
At this moment, we may have a ``\emph{partial match}'' between the pattern and the aligned text substring.
If we have performed the KMP-shift with $\kmpshift[j] < j-1$, then we have a match between the nonempty prefixes of the pattern and the aligned text substring of length $j-\kmpshift[j] -1$.
In this case, we go to the KMP-phase, where we simply perform the KMP algorithm.
The KMP-phase prevents the character comparison position from returning to the left and guarantees the linear time behavior of our algorithm.
If we have no partial match, we return to the Alignment-phase.

\begin{theorem}
The worst-case time complexity of the \distq{} algorithm is $O((n+m)q)$.
\end{theorem}
\begin{proof}
Since the proposed algorithm uses Fact~\ref{fact:kmp} on the KMP algorithm to prevent the character comparison position from going back to the left, the number of character comparisons is at most $2n-m$ times like the KMP algorithm.
In addition, the hash value of $q$-gram is calculated to perform the shift using $\shift$.
Since the hash value calculation requires $O(q)$ time and it is calculated at the maximum of $n-q+1$ places in the text, the hash value calculation takes $O(nq)$ time in total.
Therefore, the worst-case time complexity of the searching phase is $O(nq)$.
In the preprocessing, $O(mq)$ time is required to calculate the hash value of $q$-gram at $m-q+1$ locations.
\end{proof}

\begin{example}\label{ex:distq}
Let $P={\tt abaabbaaa}$.
The shifting functions $\darray$, $\kmpshift$, and $\shift$ are shown below.
The hash values are calculated by treating each character as its ASCII value, e.g.\ $\tt{a}$ is calculated as 97.
\begin{center}
	\begin{tabular}[h]{|c||c|c|c|c|c|c|c|c|c|c|}
		\hline
		$j$ & 1 & 2 & 3 & 4 & 5 & 6 & 7 & 8 & 9 & 10 \\ \hline
		$P$ & \tt{a} & \tt{b} & \tt{a} & \tt{a} & \tt{b} & \tt{b} & \tt{a} & \tt{a} & \tt{a} & \\ \hline
		$\darray$ & - & - & 1 & 2 & 3 & 4 & 5 & 4 & 7 & \\ \hline
		$\kmpshift$ & 1 & 1 & 3 & 2 & 4 & 3 & 7 & 6 & 7 & 8 \\ \hline
	\end{tabular}
\end{center}
\begin{center}
	\begin{tabular}[h]{|c||c|c|c|c|c|c|c|}
		\hline
		$x$    & \tt{aba} & \tt{baa} & \tt{aab} & \tt{abb} & \tt{bba} & \tt{aaa} & others \\ \hline
		$h(x)$ & 2041 & 2053 & 2038 & 2042 & 2057 & 2037 & \\ \hline
		$\shift[h(x)]$ & 6 & 1 & 4 & 3 & 2 & 0 & 7 \\ \hline
	\end{tabular}
\end{center}

Figure~\ref{fig:ex_distq} illustrates an example run of the DIST$q$ algorithm ($q=3$) for finding $P={\tt abaabbaaa}$ in $T={\tt abbaabbaababbabbaaabaabaabbaaa}$.
\begin{description}
\item[Attempt 1] We shift the pattern by one character for $\shift[h(T[7:9])] = \linebreak[4] \shift[h({\tt baa})] =  \shift[2053] = 1$. Since the position of the $q$-gram aligned by this shift is 8, $pos$ is updated to 8.
\item[Attempt 2] We check whether the first character of the pattern matches the corresponding character of the text. Finding $P[1] \neq T[2]$, the pattern is shifted by $\darray[pos] = \darray[8] = 4$.
\item[Attempt 3] We shift the pattern by $\shift[h(T[12:14])] = \shift[h(\tt{bba})] = 2$ and update $pos$ to 7.
\item[Attempt 4] We check whether the first character of the pattern matches the corresponding character of the text. From $P[1]=T[8] $, we compare the characters of $P[2:9]$ and $T[9:16]$ from left to right. Since $P[2] \neq T[9]$, the pattern is shifted by $\kmpshift[2]$ or $\darray[pos] = \darray[7]$. From $\kmpshift[2] = 1$, $\darray[7] = 5$, $\darray[pos] \geq 2 - 1$ and $\darray[pos] \geq \kmpshift[2]$ are satisfied. Therefore, we shift the pattern by $\darray[pos] = \darray[7] = 5$.
\item[Attempt 5] We shift the pattern by $\shift[h(T[19:21])] = \shift[h({\tt aba})] = \linebreak[4] \shift[2041] = 6$ and update $pos$ to 3.
\item[Attempt 6] We check whether the first character of the pattern matches the corresponding character of the text. By $P[1]=T[19]$, the characters of $P[2:9]$ and $T[20:27]$ are compared from left to right. Since $P[6] \neq T[24]$, the pattern is shifted by $\kmpshift[6]$ or $\darray[pos] = \darray[3]$. By $\kmpshift[6] = 3 > \darray[3] = 1$, the pattern is shifted by $\kmpshift[6] = 3$.
\item[Attempt 7] Attempt~6 shows that $P[1:2] = T[22:23]$, that is, there is a partial match, so we continue the comparison of $T[24:30]$ and $P[3:9]$. Since $T[24:30]=P[3:9]$, the pattern occurrence position 22 is reported.
\end{description}
\end{example}

\begin{figure}[t]
	\centering
		\scalebox{0.9}{
			\begin{tabular}{p{36pt}p{6pt}p{0pt}p{0pt}p{0pt}p{0pt}p{0pt}p{0pt}p{0pt}p{0pt}p{0pt}p{0pt}p{0pt}p{0pt}p{0pt}p{0pt}p{0pt}p{0pt}p{0pt}p{0pt}p{0pt}p{0pt}p{0pt}p{0pt}p{0pt}p{0pt}p{0pt}p{0pt}p{0pt}p{0pt}p{0pt}p{0pt}p{0pt}p{0pt}p{0pt}p{0pt}}
				& & 1 & 2 & 3 & 4 & 5 & 6 & 7 & 8 & 9 & 10 & 11 & 12 & 13 & 14 & 15 & 16 & 17 & 18 & 19 & 20 & 21 & 22 & 23 & 24 & 25 & 26 & 27 & 28 & 29 & 30 \\
				& $T:$ & \moji{a} & \moji{b} & \moji{b} & \moji{a} & \moji{a} & \moji{b} & \moji{\underline{b}} & \moji{\underline{a}} & \moji{\underline{a}} & \moji{b} & \moji{a} & \moji{\underline{b}} & \moji{\underline{b}} & \moji{\underline{a}} & \moji{b} & \moji{b} & \moji{a} & \moji{a} & \moji{\underline{a}} & \moji{\underline{b}} & \moji{\underline{a}} & \moji{a} & \moji{b} & \moji{a} & \moji{a} & \moji{b} & \moji{b} & \moji{a} & \moji{a} & \moji{a} \\[-2pt]
				
				\\[-6pt]
				{\small Attempt~1} & $P:$ & \moji{a} & \moji{b} & \moji{a} & \moji{a} & \moji{b} & \moji{\underline{b}} & \moji{\underline{a}} & \moji{\underline{a}} & \moji{a} \\[-2pt]
				
				&&& $\batu_{1}$ \\[-6pt]
				{\small Attempt~2} & $P:$ && \moji{a} & \moji{b} & \moji{a} & \moji{a} & \moji{b} & \moji{b} & \moji{a} & \moji{a} & \moji{a} \\[-2pt]
				
				\\[-6pt]
				{\small Attempt~3} & $P:$ &&&&&& \moji{a} & \moji{b} & \moji{a} & \moji{a} & \moji{\underline{b}} & \moji{\underline{b}} & \moji{\underline{a}} & \moji{a} & \moji{a} \\[-2pt]
				
				&&&&&&&&& $\maru_{1}$ & $\batu_{2}$ \\[-6pt]
				{\small Attempt~4} & $P:$ &&&&&&&& \moji{a} & \moji{b} & \moji{a} & \moji{a} & \moji{b} & \moji{b} & \moji{a} & \moji{a} & \moji{a} \\[-2pt]
				
				\\[-6pt]
				{\small Attempt~5} & $P:$ &&&&&&&&&&&&& \moji{a} & \moji{b} & \moji{a} & \moji{a} & \moji{b} & \moji{b} & \moji{a} & \moji{a} & \moji{a} \\[-2pt]
				
				&&&&&&&&&&&&&&&&&&&& $\maru_{1}$ & $\maru_{2}$ & $\maru_{3}$ & $\maru_{4}$ & $\maru_{5}$ & $\batu_{6}$ \\[-6pt]
				{\small Attempt~6} & $P:$ &&&&&&&&&&&&&&&&&&& \moji{a} & \moji{b} & \moji{a} & \moji{a} & \moji{b} & \moji{b} & \moji{a} & \moji{a} & \moji{a} \\[-2pt]
				
				&&&&&&&&&&&&&&&&&&&&&&& $\pmaru$ & $\pmaru$ & $\maru_{1}$ & $\maru_{2}$ & $\maru_{3}$ & $\maru_{4}$ & $\maru_{5}$ & $\maru_{6}$ & $\maru_{7}$ \\[-6pt]
				{\small Attempt~7} & $P:$ &&&&&&&&&&&&&&&&&&&&&& \moji{a} & \moji{b} & \moji{a} & \moji{a} & \moji{b} & \moji{b} & \moji{a} & \moji{a} & \moji{a} \\[-2pt]
		\end{tabular}}
		\caption{An example run of the DIST$q$ algorithm for a pattern $P={\tt abaabbaaa}$ and a text $T={\tt abbaabbaababbabbaaabaabaabbaaa}$. For each alignment of the pattern, $\maru$ and $\batu$ indicate a match and a mismatch between the text and the pattern, respectively. The character with $\pmaru$ is known to match the character at the corresponding position in the text without comparison. Subscript numbers show the order of character comparisons in each attempt.}
		\label{fig:ex_distq}
\end{figure}

\subsection{LDIST$q$ algorithm}\label{sec:ldistq}
The LDIST$q$ algorithm modifies the DIST$q$ algorithm so that the worst-case time complexity is independent of $q$.
In the DIST$q$ algorithm, if strings such as $T={\tt a}^n$ and $P={\tt ba}^{m-1}$ are given, $O(nq)$ time is required for searching phase because the hash values of each $q$-gram are calculated in the text.
Since the hash function $h$ defined in Section~\ref{sec:distq} is a rolling hash, if the hash value of $w[i:i+q-1]$ has already been obtained for a string $w$, the hash value of $w[i+1:i+q]$ can be computed in constant time by $h(w[i+1:i+q])=(4 \cdot (h(w[i:i+q-1])-4^{q-1} \cdot w[i]) + w[i+q]) {\rm \;mod\;} 2^{16}$.
The LDIST$q$ algorithm modifies Line~\ref{ln:dist_hash} of Algorithm~\ref{alg:dist_search_algorithm} so that 
we calculate the hash value of the $q$-gram using the previously calculated value of the other $q$-gram in the incremental way, if they overlap.
Similarly, the time complexity of the preprocessing phase can be reduced.

\begin{theorem}
The worst-case time complexity of the \ldistq{} algorithm is $O(n+m)$.
\end{theorem}

\begin{proof}
Like the DIST$q$ algorithm, we compare characters at most $2n-m$ times.
To calculate the hash value of a $q$-gram, if it is overlapped with the $q$-gram for which the hash value has been calculated one step before, the incremental update is performed using the rolling hash.
Therefore, the calculation of the hash value of $q$-grams takes $O(n)$ time in total.
Thus, the worst-case time complexity of matching is $O(n)$.
Calculating the hash values of $q$-grams in the preprocess is performed in the same way, so it is done in $O(m)$ time.
\end{proof}

\section{Experiments}\label{sec:experiments}
In this section, we compare the execution times of the proposed algorithms with the existing algorithms listed below, where algorithms that run in linear time in the input string size are marked with $\star$.
\begin{itemize}
	\item BNDM$q$~\cite{Navarro1998bndmq}: Backward Nondeterministic DAWG Matching algorithm using $q$-grams with $q = 2,4 {\rm \;and\;} 6$, 
	\item SBNDM$q$~\cite{Allauzen1999sbndm}: Simplified version of the Backward Nondeterministic DAWG Matching algorithm using $q$-grams with $q = 2,4,6 {\rm \;and\;} 8$,
	\item KBNDM~\cite{Cantone2012kbndm}: Factorized variant of the BNDM algorithm,
	\item BSDM$q$~\cite{Faro2012bsdmq}: Backward SNR DAWG Matching algorithm using condensed alphabets with groups of $q$ characters with $1 \leq q \leq 8$,
	\item$\star$\,FJS~\cite{Franek2007fjs}: Franek-Jennings-Smyth algorithm, 
	\item$\star$\,FJS$+$~\cite{Kobayashi2019improvedfjs}: Modification of the FJS algorithm, 
	\item HASH$q$~\cite{Lecroq2007hashq}: Hashing algorithm using $q$-grams with $2 \leq q \leq 8$ (see Section~\ref{sec:hashq}),
	\item FS-$w$~\cite{Faro2012multi}: Multiple Windows version of the Fast Search algorithm~\cite{Cantone2005fast} implemented using $w$ sliding windows with $w = 1,2,4,6 {\rm \;and\;} 8$,
	\item IOM~\cite{Cantone2014occurrence}: Improved Occurrence Matcher, 
	\item WOM~\cite{Cantone2014occurrence}: Worst Occurrence Matcher, 
	\item SKIP$q$~\cite{Faro2016skipq}: Skip-Search algorithm using $q$-grams with $2 \leq q \leq 8$, 
	\item WFR$q$~\cite{Cantone2019wfr}: Weak-Factors-Recognition algorithm implemented with a $q$-chained loop with $2 \leq q \leq 8$, 
	\item$\star$\,LWFR$q$~\cite{Cantone2019wfr}: Linear-Weak-Factors-Recognition algorithm implemented with a $q$-chained loop with $2 \leq q \leq 8$, 
	\item$\star$\,DIST$q$: Our algorithm proposed in Section~\ref{sec:distq} (Algorithm~\ref{alg:dist_search_algorithm}) with $1 \leq q \leq 8$, 
	\item$\star$\,LDIST$q$: Our algorithm proposed in Section~\ref{sec:ldistq} with $1 \leq q \leq 8$. 
\end{itemize}

\newcolumntype{A}{D{+}{}{5}}
\newcommand{\winner}[2]{\underline{\bf{#1}}+\underline{\bf\ensuremath{.#2}}}

\begin{table}[b!]
	\caption{Genome sequence ($\sigma = 4$, $n=4641652$)}
	\label{tb:genome}
	\centering
	{\renewcommand{\arraystretch}{1.0}\setlength{\tabcolsep}{1mm}
		\scalebox{0.75}{
			\begin{tabular}{|l|AAAAAAAAAA|}
				\hline
				\multicolumn{1}{|c|}{$m$} & \multicolumn{1}{c}{2} & \multicolumn{1}{c}{4} & \multicolumn{1}{c}{8} & \multicolumn{1}{c}{16} & \multicolumn{1}{c}{32} & \multicolumn{1}{c}{64} & \multicolumn{1}{c}{128} & \multicolumn{1}{c}{256} & \multicolumn{1}{c}{512} & \multicolumn{1}{c|}{1024} \\
				\hline
BNDM$q$  & 218+.34^{(2)}& 174+.55^{(2)}& 84+.14^{(4)}& 64+.89^{(6)}& 60+.39^{(6)}& 61+.07^{(6)}& 60+.85^{(6)}& 62+.17^{(6)}& 61+.20^{(6)}& 60+.16^{(6)} \\
SBNDM$q$ & \winner{179}{90^{(2)}}& 154+.20^{(2)}& 80+.94^{(4)}& 73+.29^{(6)}& 57+.10^{(8)}& 61+.01^{(6)}& 61+.74^{(6)}& 61+.20^{(6)}& 61+.67^{(6)}& 60+.80^{(6)} \\
KBNDM  & 311+.78& 201+.99& 150+.15& 113+.84& 83+.23& 67+.83& 75+.65& 75+.39& 76+.58& 74+.72 \\
BSDM$q$  & 195+.38^{(2)}& \winner{118}{86^{(3)}}& 84+.23^{(5)}& 63+.03^{(6)}& 61+.26^{(6)}& 58+.75^{(6)}& 57+.50^{(7)}& 56+.99^{(6)}& 57+.01^{(6)}& 56+.58^{(6)} \\
$\star$\,FJS    & 407+.02& 353+.60& 311+.96& 279+.13& 308+.42& 297+.00& 266+.12& 317+.79& 317+.34& 296+.19 \\
$\star$\,FJS$+$ & 388+.44& 296+.70& 203+.03& 171+.17& 149+.52& 136+.59& 128+.55& 130+.39& 122+.51& 112+.99 \\
HASH$q$  & 571+.44^{(2)}& 272+.86^{(3)}& 126+.00^{(3)}& 88+.12^{(3)}& 68+.22^{(3)}& 58+.84^{(6)}& 55+.09^{(6)}& 59+.48^{(7)}& 57+.34^{(7)}& 57+.96^{(7)} \\
FS-$w$   & 332+.32^{(4)}& 245+.99^{(4)}& 184+.72^{(4)}& 158+.72^{(4)}& 143+.79^{(6)}& 125+.05^{(4)}& 123+.52^{(6)}& 117+.90^{(6)}& 108+.03^{(6)}& 100+.72^{(6)} \\
IOM    & 377+.25& 275+.36& 215+.72& 220+.97& 219+.86& 218+.12& 210+.61& 221+.31& 230+.15& 211+.69 \\
WOM    & 381+.54& 301+.46& 220+.34& 182+.30& 166+.27& 143+.24& 136+.20& 133+.75& 127+.40& 114+.74 \\
SKIP$q$  & 250+.89^{(2)}& 136+.18^{(3)}& 91+.51^{(4)}& 63+.96^{(6)}& 56+.79^{(7)}& 53+.09^{(7)}& 52+.10^{(7)}& 57+.12^{(7)}& 56+.97^{(8)}& 58+.00^{(6)} \\
WFR$q$   & 219+.50^{(2)}& 168+.39^{(2)}& 88+.86^{(4)}& 65+.82^{(4)}& 57+.19^{(8)}& 55+.12^{(2)}& 51+.77^{(3)}& 50+.04^{(3)}& 55+.02^{(6)}& \winner{54}{64^{(8)}} \\
$\star$\,LWFR$q$  & 216+.10^{(2)}& 173+.48^{(3)}& 88+.71^{(4)}& 60+.75^{(5)}& \winner{53}{84^{(5)}}& \winner{50}{48^{(6)}}& \winner{49}{65^{(8)}}& 48+.71^{(6)}& 54+.90^{(6)}& 54+.97^{(7)} \\
$\star$\,DIST$q$  & 186+.10^{(2)}& 125+.44^{(3)}& \winner{78}{56^{(4)}}& \winner{60}{48^{(5)}}& 55+.21^{(6)}& 52+.05^{(7)}& 51+.26^{(8)}& 50+.44^{(8)}& 54+.81^{(7)}& 55+.58^{(8)} \\
$\star$\,LDIST$q$ & 295+.55^{(2)}& 181+.99^{(3)}& 86+.58^{(4)}& 65+.29^{(6)}& 56+.74^{(6)}& 52+.31^{(6)}& 50+.39^{(6)}& \winner{48}{70^{(7)}}& \winner{54}{33^{(4)}}& 55+.22^{(7)} \\
				\hline
			\end{tabular}
		}
	}
\end{table}

\begin{table}[b!]
	\caption{English text ($\sigma = 62$, $n=4017009$)}
	\label{tb:english-ajkv}
	\centering
	{\renewcommand{\arraystretch}{1.0}\setlength{\tabcolsep}{1mm}
		\scalebox{0.75}{
			\begin{tabular}{|l|AAAAAAAAAA|}
				\hline
				\multicolumn{1}{|c|}{$m$} & \multicolumn{1}{c}{2} & \multicolumn{1}{c}{4} & \multicolumn{1}{c}{8} & \multicolumn{1}{c}{16} & \multicolumn{1}{c}{32} & \multicolumn{1}{c}{64} & \multicolumn{1}{c}{128} & \multicolumn{1}{c}{256} & \multicolumn{1}{c}{512} & \multicolumn{1}{c|}{1024} \\
				\hline
BNDM$q$  & 140+.48^{(2)}& 93+.39^{(2)}& 70+.84^{(4)}& 54+.10^{(4)}& \winner{45}{88^{(4)}}& 47+.61^{(4)}& 47+.64^{(4)}& 46+.88^{(4)}& 47+.40^{(6)}& 45+.58^{(4)} \\
SBNDM$q$ & \winner{108}{32^{(2)}}& \winner{73}{50^{(2)}}& \winner{64}{87^{(2)}}& \winner{50}{73^{(4)}}& 47+.85^{(4)}& 48+.45^{(4)}& 48+.55^{(4)}& 47+.56^{(4)}& 47+.08^{(4)}& 45+.98^{(4)} \\
KBNDM  & 192+.76& 126+.56& 92+.03& 71+.04& 64+.17& 56+.79& 49+.56& 49+.09& 50+.24& 47+.68 \\
BSDM$q$  & 117+.22^{(2)}& 79+.96^{(2)}& 70+.36^{(3)}& 57+.72^{(4)}& 49+.91^{(8)}& 48+.00^{(8)}& 44+.99^{(8)}& 43+.62^{(8)}& 43+.06^{(8)}& 42+.70^{(8)} \\
$\star$\,FJS    & 192+.52& 158+.61& 106+.40& 89+.04& 78+.43& 68+.14& 64+.80& 61+.15& 60+.48& 55+.38 \\
$\star$\,FJS$+$ & 196+.88& 155+.86& 101+.77& 77+.74& 67+.40& 60+.31& 54+.77& 52+.44& 51+.62& 47+.74 \\
HASH$q$  & 312+.91^{(2)}& 218+.95^{(2)}& 107+.38^{(2)}& 70+.85^{(2)}& 56+.90^{(3)}& 49+.33^{(6)}& 46+.63^{(5)}& 45+.27^{(8)}& 44+.77^{(3)}& 42+.98^{(7)} \\
FS-$w$   & 129+.50^{(6)}& 104+.45^{(6)}& 71+.57^{(6)}& 61+.08^{(4)}& 54+.19^{(6)}& 49+.53^{(4)}& 48+.95^{(6)}& 47+.02^{(4)}& 46+.96^{(6)}& 43+.00^{(4)} \\
IOM    & 193+.37& 148+.51& 104+.36& 86+.22& 75+.09& 69+.08& 63+.63& 60+.62& 58+.83& 51+.76 \\
WOM    & 199+.23& 153+.81& 112+.45& 86+.61& 75+.26& 62+.80& 60+.54& 62+.74& 58+.91& 55+.65 \\
SKIP$q$  & 161+.38^{(2)}& 100+.52^{(2)}& 72+.33^{(3)}& 55+.71^{(4)}& 49+.06^{(4)}& 48+.73^{(4)}& 49+.87^{(4)}& 48+.81^{(8)}& 45+.75^{(2)}& 45+.32^{(2)} \\
WFR$q$   & 137+.40^{(2)}& 85+.22^{(2)}& 68+.88^{(2)}& 52+.85^{(5)}& 46+.67^{(5)}& \winner{44}{39^{(8)}}& \winner{42}{09^{(8)}}& 41+.66^{(5)}& 41+.65^{(6)}& 40+.53^{(2)} \\
$\star$\,LWFR$q$  & 121+.08^{(2)}& 85+.47^{(2)}& 70+.38^{(2)}& 53+.83^{(3)}& 47+.89^{(5)}& 44+.49^{(5)}& 42+.38^{(6)}& \winner{41}{09^{(8)}}& \winner{41}{23^{(8)}}& \winner{40}{30^{(8)}} \\
$\star$\,DIST$q$  & 115+.61^{(2)}& 80+.14^{(2)}& 65+.84^{(3)}& 52+.25^{(4)}& 48+.13^{(4)}& 46+.26^{(4)}& 43+.32^{(4)}& 42+.84^{(8)}& 42+.98^{(4)}& 41+.47^{(7)} \\
$\star$\,LDIST$q$ & 229+.62^{(2)}& 102+.15^{(2)}& 69+.70^{(3)}& 55+.34^{(4)}& 49+.46^{(5)}& 45+.93^{(5)}& 44+.37^{(3)}& 43+.15^{(4)}& 42+.21^{(5)}& 41+.11^{(7)} \\
				\hline
			\end{tabular}
		}
	}
\end{table}

\newcommand{\nr}[1]{\scalebox{1}[1]{#1}}
\newcommand{\nrr}[1]{\scalebox{1}[1]{#1}}
\begin{table}[b!]
	\caption{Fibonacci string ($\sigma=2$, $n=2178309$)}
	\label{tb:fibonacci}
	\centering
	{\renewcommand{\arraystretch}{1.0}\setlength{\tabcolsep}{1mm}
		\scalebox{0.75}{
			\begin{tabular}{|l|AAAAAAAAAA|}
				\hline
				\multicolumn{1}{|c|}{$m$} & \multicolumn{1}{c}{2} & \multicolumn{1}{c}{4} & \multicolumn{1}{c}{8} & \multicolumn{1}{c}{16} & \multicolumn{1}{c}{32} & \multicolumn{1}{c}{64} & \multicolumn{1}{c}{128} & \multicolumn{1}{c}{256} & \multicolumn{1}{c}{512} & \multicolumn{1}{c|}{1024} \\
				\hline
BNDM$q$  & 343+.25^{(2)}& 308+.44^{(2)}& 283+.26^{(4)}& 257+.64^{(6)}& 233+.94^{(6)}& 285+.63^{(4)}& 284+.37^{(4)}& 293+.00^{(4)}& 307+.82^{(4)}& 315+.47^{(4)} \\
SBNDM$q$ & \winner{286}{02^{(2)}}& \!\winner{292}{15^{(2)}}& 272+.98^{(4)}& 276+.35^{(6)}& 306+.42^{(6)}& 372+.03^{(6)}& 432+.53^{(6)}& 493+.20^{(6)}& 546+.94^{(6)}& 602+.09^{(6)} \\
KBNDM  & 541+.70& 405+.78& 411+.08& 422+.85& 382+.25& 402+.45& 425+.60& 437+.67& 461+.07& 451+.12 \\
BSDM$q$  & 482+.47^{(2)}& 500+.43^{(3)}& 397+.29^{(5)}& 362+.52^{(8)}& 330+.76^{(6)}& 736+.89^{(1)}& 766+.57^{(1)}& 782+.98^{(1)}& 790+.26^{(1)}& 508+.80^{(3)} \\
$\star$\,FJS    & 402+.44& 362+.23& 276+.97& 237+.87& 218+.38& \winner{206}{07}& 203+.86& 202+.94& 196+.49& 194+.01 \\
$\star$\,FJS$+$ & 456+.20& 396+.04& 335+.70& 319+.93& 295+.64& 300+.36& 296+.48& 295+.37& 288+.56& 289+.00 \\
HASH$q$  & 645+.48^{(2)}& 406+.70^{(2)}& \!\winner{257}{69^{(4)}}& 251+.91^{(7)}& 279+.81^{(7)}& 344+.99^{(7)}& 415+.16^{(7)}& 470+.71^{(7)}& 514+.05^{(7)}& 579+.57^{(7)} \\
FS-$w$   & 383+.23^{(1)}& 396+.23^{(1)}& 347+.72^{(1)}& 289+.15^{(1)}& 253+.36^{(1)}& 246+.82^{(1)}& 248+.73^{(1)}& 235+.66^{(1)}& 235+.35^{(1)}& 230+.61^{(1)} \\
IOM    & 381+.92& 414+.42& 453+.54& 497+.84& 543+.93& 641+.13& 751+.42& 839+.92& 899+.19& 1019+.59 \\
WOM    & 552+.38& 555+.43& 564+.67& 617+.93& 664+.47& 732+.05& 852+.06& 926+.35& 1036+.31& 1126+.17 \\
SKIP$q$  & 470+.93^{(2)}& 394+.09^{(2)}& 332+.66^{(5)}& 336+.16^{(8)}& 374+.91^{(8)}& 464+.23^{(3)}& 460+.54^{(3)}& 450+.18^{(3)}& 451+.05^{(3)}& 464+.13^{(3)} \\
WFR$q$   & 442+.32^{(2)}& 497+.95^{(3)}& 528+.45^{(3)}& 652+.48^{(6)}& \nr{2132}+.38^{(7)}& \nr{3762}+.19^{(8)}& \nr{6762}+.67^{(8)}& \nrr{12624}+.63^{(8)}& \nrr{24416}+.65^{(8)}& \nrr{48596}+.02^{(8)} \\
$\star$\,LWFR$q$  & 552+.64^{(2)}& 504+.77^{(3)}& 428+.34^{(5)}& 342+.80^{(7)}& 297+.07^{(7)}& 304+.57^{(2)}& 274+.47^{(6)}& 265+.28^{(6)}& 258+.06^{(6)}& 254+.92^{(6)} \\
$\star$\,DIST$q$  & 438+.96^{(1)}& 359+.56^{(2)}& 293+.80^{(2)}& \!\winner{235}{61^{(2)}}& \winner{213}{07^{(7)}}& 206+.92^{(2)}& 201+.18^{(8)}& 196+.56^{(5)}& 194+.86^{(4)}& 193+.62^{(5)} \\
$\star$\,LDIST$q$ & 565+.13^{(2)}& 412+.15^{(3)}& 300+.98^{(4)}& 245+.38^{(7)}& 215+.89^{(8)}& 208+.40^{(7)}& \winner{200}{45^{(4)}}& \winner{193}{74^{(4)}}& \winner{190}{85^{(3)}}& \winner{193}{48^{(8)}} \\
				\hline
			\end{tabular}
		}
	}
\end{table}

\begin{table}[t!]
	\caption{Texts with frequent pattern occurrences ($\sigma = 4$, $n = 4000000$, $m=8$)}
	\label{tb:occ-m-8-sigma-4}
	\centering
	{\renewcommand{\arraystretch}{1.0}\setlength{\tabcolsep}{1mm}
		\scalebox{0.7}{
			\begin{tabular}{|l|AAAAAAAAAAAA|}
				\hline
				\multicolumn{1}{|c|}{$occ$} & \multicolumn{1}{c}{0} & \multicolumn{1}{c}{128} & \multicolumn{1}{c}{256} & \multicolumn{1}{c}{512} & \multicolumn{1}{c}{1024} & \multicolumn{1}{c}{2048} & \multicolumn{1}{c}{4096} & \multicolumn{1}{c}{8192} & \multicolumn{1}{c}{16384} & \multicolumn{1}{c}{32768} & \multicolumn{1}{c}{65536} & \multicolumn{1}{c|}{131072} \\
				\hline
BNDM$q$  & 73+.53^{(4)}& 72+.28^{(4)}& 72+.87^{(4)}& 73+.89^{(4)}& 74+.39^{(4)}& 75+.75^{(4)}& 77+.60^{(4)}& 81+.63^{(4)}& 82+.67^{(4)}& 88+.98^{(4)}& 108+.56^{(4)}& 146+.13^{(6)} \\
SBNDM$q$ & \winner{68}{58^{(4)}}& \winner{68}{92^{(4)}}& 71+.17^{(4)}& 77+.26^{(4)}& 81+.04^{(4)}& 80+.76^{(4)}& 78+.33^{(4)}& 82+.82^{(4)}& 91+.01^{(4)}& 96+.83^{(4)}& 111+.75^{(4)}& 140+.21^{(4)} \\
KBNDM  & 127+.73& 128+.41& 128+.41& 126+.92& 128+.74& 131+.17& 129+.74& 135+.70& 140+.99& 152+.46& 164+.59& 186+.32 \\
BSDM$q$  & 69+.85^{(4)}& 69+.04^{(4)}& \winner{68}{85^{(4)}}& \winner{69}{19^{(4)}}& \winner{70}{63^{(4)}}& \winner{72}{00^{(4)}}& \winner{72}{48^{(4)}}& \winner{76}{25^{(4)}}& \winner{79}{95^{(4)}}& 89+.91^{(4)}& 110+.22^{(4)}& 143+.91^{(4)} \\
$\star$\,FJS    & 246+.26& 253+.35& 263+.94& 247+.73& 255+.94& 276+.44& 262+.83& 256+.36& 251+.33& 269+.39& 259+.93& 251+.06 \\
$\star$\,FJS$+$ & 174+.38& 173+.86& 180+.46& 173+.05& 178+.65& 182+.98& 177+.56& 181+.17& 182+.48& 190+.71& 197+.69& 199+.60 \\
HASH$q$  & 121+.09^{(3)}& 121+.36^{(3)}& 122+.64^{(3)}& 122+.50^{(3)}& 119+.49^{(3)}& 123+.98^{(3)}& 124+.31^{(3)}& 124+.93^{(3)}& 126+.81^{(3)}& 128+.94^{(3)}& 143+.75^{(3)}& 153+.08^{(4)} \\
FS-$w$   & 154+.89^{(4)}& 155+.56^{(4)}& 165+.18^{(4)}& 156+.67^{(4)}& 159+.03^{(4)}& 166+.44^{(4)}& 165+.62^{(4)}& 160+.46^{(4)}& 167+.46^{(4)}& 178+.53^{(2)}& 185+.86^{(2)}& 191+.38^{(2)} \\
IOM    & 185+.31& 181+.07& 195+.53& 187+.98& 194+.18& 200+.99& 195+.98& 195+.89& 197+.67& 202+.61& 214+.90& 209+.52 \\
WOM    & 192+.59& 192+.28& 207+.57& 189+.21& 196+.02& 203+.86& 196+.98& 199+.79& 203+.59& 215+.79& 219+.94& 230+.59 \\
SKIP$q$  & 79+.32^{(4)}& 77+.14^{(4)}& 79+.19^{(4)}& 82+.08^{(4)}& 81+.82^{(4)}& 82+.55^{(4)}& 83+.83^{(4)}& 87+.08^{(4)}& 89+.82^{(4)}& 93+.09^{(4)}& 106+.41^{(4)}& 126+.22^{(3)} \\
WFR$q$   & 76+.68^{(4)}& 76+.38^{(4)}& 77+.63^{(4)}& 83+.21^{(4)}& 80+.93^{(4)}& 81+.42^{(4)}& 83+.86^{(4)}& 88+.55^{(4)}& 93+.16^{(4)}& 106+.07^{(4)}& 123+.90^{(4)}& 164+.77^{(4)} \\
$\star$\,LWFR$q$  & 76+.99^{(4)}& 76+.33^{(4)}& 78+.83^{(4)}& 77+.57^{(4)}& 78+.74^{(4)}& 76+.46^{(4)}& 84+.65^{(4)}& 89+.87^{(4)}& 96+.17^{(4)}& 108+.67^{(4)}& 129+.35^{(3)}& 168+.70^{(3)} \\
$\star$\,DIST$q$  & 69+.67^{(4)}& 73+.38^{(4)}& 74+.35^{(4)}& 74+.31^{(4)}& 75+.12^{(4)}& 74+.96^{(4)}& 77+.21^{(4)}& 77+.56^{(4)}& 80+.09^{(4)}& \winner{86}{25^{(4)}}& \winner{101}{12^{(4)}}& \winner{120}{98^{(3)}} \\
$\star$\,LDIST$q$ & 75+.82^{(4)}& 74+.45^{(4)}& 74+.89^{(4)}& 76+.77^{(4)}& 74+.62^{(4)}& 75+.47^{(4)}& 77+.10^{(4)}& 80+.18^{(4)}& 83+.40^{(4)}& 88+.86^{(4)}& 103+.73^{(4)}& 122+.27^{(4)} \\
				\hline
			\end{tabular}
		}
	}
\end{table}

\begin{table}[t!]
	\caption{Texts with frequent pattern occurrences ($\sigma = 95$, $n = 4000000$, $m=8$)}
	\label{tb:occ-m-8-sigma-95}
	\centering
	{\renewcommand{\arraystretch}{1.0}\setlength{\tabcolsep}{1mm}
		\scalebox{0.7}{
			\begin{tabular}{|l|AAAAAAAAAAAA|}
				\hline
				\multicolumn{1}{|c|}{$occ$} & \multicolumn{1}{c}{0} & \multicolumn{1}{c}{128} & \multicolumn{1}{c}{256} & \multicolumn{1}{c}{512} & \multicolumn{1}{c}{1024} & \multicolumn{1}{c}{2048} & \multicolumn{1}{c}{4096} & \multicolumn{1}{c}{8192} & \multicolumn{1}{c}{16384} & \multicolumn{1}{c}{32768} & \multicolumn{1}{c}{65536} & \multicolumn{1}{c|}{131072} \\
				\hline
BNDM$q$  & 56+.58^{(2)}& 55+.41^{(2)}& 56+.26^{(2)}& 56+.89^{(2)}& 56+.68^{(2)}& 57+.42^{(2)}& 58+.80^{(2)}& 60+.72^{(2)}& 67+.37^{(2)}& 74+.57^{(2)}& 98+.52^{(2)}& 125+.65^{(2)} \\
SBNDM$q$ & \winner{52}{76^{(2)}}& \winner{52}{37^{(2)}}& \winner{53}{09^{(2)}}& \winner{53}{23^{(2)}}& 55+.56^{(2)}& \winner{52}{82^{(2)}}& 56+.41^{(2)}& \winner{56}{50^{(2)}}& 61+.58^{(2)}& 69+.96^{(2)}& 91+.72^{(2)}& 112+.41^{(2)} \\
KBNDM  & 81+.98& 80+.22& 77+.61& 79+.42& 79+.55& 81+.88& 82+.61& 82+.25& 89+.19& 98+.65& 111+.59& 142+.65 \\
BSDM$q$  & 54+.89^{(2)}& 56+.07^{(2)}& 54+.55^{(2)}& 55+.99^{(2)}& 55+.78^{(2)}& 56+.78^{(2)}& 58+.38^{(2)}& 61+.48^{(2)}& 66+.76^{(2)}& 76+.66^{(3)}& 96+.79^{(3)}& 131+.76^{(3)} \\
$\star$\,FJS    & 81+.93& 81+.41& 81+.14& 81+.83& 82+.04& 82+.16& 82+.10& 84+.30& 86+.02& 90+.72& 100+.41& 111+.55 \\
$\star$\,FJS$+$ & 80+.26& 84+.27& 81+.05& 80+.58& 80+.21& 80+.84& 83+.52& 85+.66& 87+.40& 94+.03& 102+.16& 112+.23 \\
HASH$q$  & 101+.34^{(2)}& 103+.34^{(2)}& 102+.15^{(2)}& 105+.02^{(2)}& 101+.55^{(2)}& 104+.07^{(2)}& 105+.84^{(2)}& 107+.46^{(2)}& 107+.79^{(2)}& 112+.15^{(2)}& 118+.52^{(2)}& 126+.53^{(2)} \\
FS-$w$   & 52+.77^{(8)}& 52+.42^{(6)}& 53+.51^{(6)}& 53+.74^{(6)}& \winner{53}{07^{(6)}}& 53+.95^{(6)}& \winner{55}{33^{(6)}}& 57+.61^{(6)}& \winner{61}{36^{(6)}}& 68+.86^{(6)}& 81+.66^{(4)}& 102+.11^{(2)} \\
IOM    & 82+.07& 83+.86& 84+.89& 85+.58& 82+.64& 82+.45& 83+.90& 86+.98& 87+.28& 91+.63& 99+.30& 106+.37 \\
WOM    & 84+.37& 84+.42& 86+.39& 85+.05& 85+.36& 85+.35& 86+.20& 86+.74& 90+.22& 93+.31& 105+.36& 114+.68 \\
SKIP$q$  & 60+.93^{(2)}& 59+.34^{(2)}& 63+.66^{(3)}& 62+.97^{(3)}& 62+.11^{(2)}& 62+.42^{(2)}& 64+.08^{(2)}& 65+.12^{(2)}& 69+.11^{(2)}& 76+.18^{(3)}& 83+.05^{(3)}& 101+.43^{(3)} \\
WFR$q$   & 59+.39^{(2)}& 59+.57^{(2)}& 59+.24^{(2)}& 60+.12^{(2)}& 61+.16^{(2)}& 55+.51^{(2)}& 58+.24^{(2)}& 62+.38^{(2)}& 68+.35^{(2)}& 81+.99^{(2)}& 104+.92^{(2)}& 139+.15^{(2)} \\
$\star$\,LWFR$q$  & 55+.99^{(2)}& 58+.63^{(2)}& 54+.16^{(2)}& 53+.75^{(2)}& 56+.99^{(2)}& 60+.87^{(2)}& 58+.45^{(2)}& 63+.39^{(2)}& 72+.48^{(2)}& 85+.61^{(2)}& 112+.47^{(3)}& 152+.54^{(3)} \\
$\star$\,DIST$q$  & 58+.60^{(2)}& 59+.30^{(2)}& 58+.94^{(2)}& 58+.46^{(2)}& 58+.47^{(3)}& 59+.91^{(2)}& 61+.55^{(2)}& 62+.29^{(2)}& 65+.21^{(2)}& \winner{65}{99^{(2)}}& \winner{77}{36^{(2)}}& \winner{95}{05^{(2)}} \\
$\star$\,LDIST$q$ & 62+.56^{(3)}& 61+.52^{(2)}& 66+.62^{(3)}& 66+.90^{(2)}& 67+.28^{(3)}& 63+.66^{(2)}& 65+.22^{(2)}& 67+.15^{(2)}& 67+.59^{(2)}& 73+.93^{(2)}& 85+.10^{(2)}& 100+.33^{(2)} \\
				\hline
			\end{tabular}
		}
	}
\end{table}

All algorithms are implemented in C language, compiled by GCC 9.2.0 with the optimization option $\mathtt{-O3}$. 
We used the implementations in SMART~\cite{Faro2016smart} for all algorithms except for the FJS, FJS$+$ and our algorithms.
The implementations of our algorithms are available at \url{https://github.com/ushitora/distq}.
We experimented with the following strings:
\begin{enumerate}
	\item Genome sequence (Table~\ref{tb:genome}): the genome sequence of \emph{E.~coli} of length $n=4641652$ with $\sigma=4$, from NCBI\footnote{\texttt{https://www.ncbi.nlm.nih.gov/genome/167?genome\_assembly\_id=161521}}. The patterns are randomly extracted from $T$ of length $m=2,4,8,16,32,64,128,256,512 \text{\;and\;} 1024$. We measured the total running time of 25 executions.
	\item English text (Table~\ref{tb:english-ajkv}): the King James version of the Bible of length $n=4017009$ with $\sigma=62$, from the Large Canterbury Corpus\footnote{\texttt{http://corpus.canterbury.ac.nz/}}~\cite{Arnold1997canterbury}. We removed the line breaks from the text. The patterns are randomly extracted from $T$ of length $m=2,4,8,16,32,64,128,256,512$ and $1024$. We measured the total running time of 25 executions.
	\item Fibonacci string (Table~\ref{tb:fibonacci}): generated by the following recurrence
	\[
	\Fib_1=\texttt{b},\;\Fib_2=\texttt{a}\;\text{ and }\;\Fib_k =\Fib_{k-1} \cdot \Fib_{k-2} \text{ for } k > 2. 
	\]
	The text is fixed to $T=\Fib_{32}$ of length $n = 2178309$. The patterns are randomly extracted from $T$ of length $m=2,4,8,16,32,64,128,256,512 \text{\;and\;} 1024$. We measured the total running time of 100 executions.
	\item Texts with frequent pattern occurrences (Tables~\ref{tb:occ-m-8-sigma-4}, \ref{tb:occ-m-8-sigma-95}): generated by intentionally embedding a lot of patterns. We embedded $\mathit{occ} = 0,128,256,512,1024,2048,4096,8192,16384,\linebreak[1]32768,65536$, and $131072$ occurrences of a pattern of length $m=8$ into a text of length $n=4000000$ over an alphabet of size $\sigma=4\mbox{ and }95$.
	More specifically, we first randomly generate a pattern and a provisional text, which may contain the pattern. Then we randomly change characters of the text until the pattern does not occur in the text. Finally we embed the pattern $\mathit{occ}$ times at random positions without overlapping. We measured the total running time of 25 executions.
	\end{enumerate}
The best performance among three trials is recorded for each experiment.
For the algorithms using parameter $q$ or $w$, we report only the best results.
The value of $q$ or $w$ giving the best performance is shown in round brackets.

Experimental results show that when the pattern is short, the SBNDM$q$ and BSDM$q$ algorithms have good performance in general.
For the genome sequence text, WFR$q$, LWFR$q$ and our algorithms are the fastest algorithms except when the pattern is very short.
On the English text, SBNDM$q$ and LWFR$q$ run fastest for short and long patterns, respectively.
On the other hand, DIST$q$ runs almost as fast as the best algorithm on both short and long patterns.
In fact, it runs faster than SBDDM$q$ and LWFR$q$ for long and short patterns, respectively.
In the experiments on the Fibonacci string, the FJS algorithm and our algorithms have shown good results as the pattern length increases.
Differently from the previous two sorts of texts, our algorithms clearly outperformed the LWFR$q$ algorithm.
Since the Fibonacci strings have many repeating structures and patterns are randomly extracted from the text, the number of occurrences of the pattern is very large in this experiment.
Therefore, we hypothesize that the efficiency of DIST$q$ algorithms does not decrease when the number of pattern occurrences is large.
We fixed the pattern length and alphabet size and prepared data with the number of pattern occurrences intentionally changed.
From the experimental result, it is found that our algorithms become more advantageous as the number of pattern occurrences increases.
The results show that the LDIST$q$ algorithm is generally slower than the DIST$q$ algorithm.
This should be due to the overhead of the process of determining whether to update the hash value difference by the rolling hash in the LDIST$q$ algorithm.

\section{Conclusion}\label{sec:conclusion}
We proposed two new algorithms for the exact string matching problem: the DIST$q$ algorithm and the LDIST$q$ algorithm.
We confirmed that our algorithms are as efficient as the state-of-the-art algorithms in many cases.
Particularly when a pattern frequently appears in a text, our algorithms outperformed existing algorithms.
The DIST$q$ algorithm runs in $O(q(n+m))$ time and the LDIST$q$ algorithm runs in $O(n+m)$ time.
Their performances were not significantly different in our experiments and rather the former ran faster than the latter in most cases, where the optimal value of $q$ was relatively small.

\bibliographystyle{plain} 
\bibliography{ref}

\end{document}